\documentclass{article}
\usepackage{graphicx}   
\usepackage{amsmath}    
\usepackage{amssymb}    
\usepackage{booktabs}   
\usepackage{float}      
\usepackage{caption}    
\usepackage{indentfirst}
\usepackage{hyperref}   
\usepackage{tabularx}
\usepackage{threeparttable}
\usepackage{natbib}
\usepackage{authblk}
\usepackage{amsthm}

\newtheorem{proposition}{Proposition}
\usepackage{placeins}
\usepackage{adjustbox}
\usepackage{lscape}
\usepackage{adjustbox}
\usepackage{array}
\usepackage{float}
\usepackage{geometry}
\title{Liquidity Premium and Investment Horizons}
\author[1]{Irene Aldridge\thanks{Email: irene.aldridge@gmail.com}}
\date{}

\begin{document}

\maketitle

\begin{abstract}
We estimate Kyle's (1985) price-impact coefficient $\lambda$ directly from daily equity order flow and test its ability to forecast the cross-section of subsequent stock returns. Using CRSP data from 2020 to 2025, we construct firm-month measures of signed order flow and two estimators of $\hat\lambda_{it}$: a within-month price-impact regression and an Amihud-style ratio. Signed order flow strongly predicts contemporaneous and one-month-ahead returns, while volume volatility predicts lower subsequent returns, consistent with widening price impact degrading price discovery. Fama-MacBeth regressions confirm that our order-flow signal carries significant cross-sectional return information after Newey--West adjustment. Theoretically, we resolve the liquidity premium puzzle of Constantinides (1986) through an adverse-selection mechanism: low order flow widens $\lambda$ and depresses prices today; subsequent normalization restores prices, generating the illiquidity premium without risk-based compensation.
\end{abstract}

\section{Introduction}
 
Equity returns exhibit return predictability tied to trading
activity that classical asset-pricing models struggle to explain. A long empirical literature documents that trading volume, order flow, and measures of illiquidity carry information about subsequent stock returns. However, this evidence has developed largely independently of the equilibrium models that explain \emph{why} trading activity should move prices in the first place. We bring these two strands together: we estimate the price-impact coefficient from \cite{Kyle1985} model of strategic informed trading directly from observable equity order flow and test its ability to forecast the cross-section of subsequent stock returns.
 
In Kyle's model, a single informed trader submits orders against a competitive, risk-neutral market maker who observes only aggregate order flow and prices the asset as a linear function of that flow. The resulting equilibrium price-impact coefficient, $\lambda$, governs how strongly a given unit of net order flow moves prices, and is determined by the ratio of noise-trader variance to fundamental uncertainty about the asset's value. We argue that this equilibrium object, estimated from daily equity trading data, predicts subsequent returns. A stock with a high $\lambda$ is one whose order flow is, in equilibrium, more informative, more costly to trade against, and more sensitive to shifts in investor demand.
 
We propose a resolution to a long-standing puzzle in this
literature: why does trading volume predict returns, and through
what mechanism does illiquidity command a premium? Existing
empirical work documents a robust association between volume,
order flow, and subsequent returns (e.g., \cite{LoWang2015}, \cite{Bajzik2021}), but does not, on its own, provide a structural account of why this association should hold or how to translate it into a tractable equilibrium pricing model. Separately, a large literature on equity liquidity premia (\cite{AmihudMendelson1986}; \cite{PastorStambaugh2003}; \cite{AcharyaPedersen2005}) establishes that illiquidity is priced, but typically treats the illiquidity premium as compensation for risk borne by long-horizon investors, leaving
open the puzzle first posed by \cite{Constantinides1986}: if
infinite-horizon investors should be indifferent to transaction
costs, why is the observed illiquidity premium so large in practice?
 
We propose that both questions have a common answer rooted in
market microstructure theory. \cite{Kyle1985}'s equilibrium implies that $\lambda$ widens precisely when the ratio of informed to noise trading is high relative to fundamental uncertainty, discouraging participation and depressing prices. As order flow normalizes, $\lambda$ narrows and prices recover. This generates a return differential between low- and high-order-flow states that requires no counterparty to knowingly bear a liquidity cost on investors' behalf. The $\lambda$ synamics further offer a resolution to the liquidity premium puzzle grounded in adverse selection rather than risk compensation.
 
We test four predictions of this framework, each tied to a formal proposition developed in Section 3. We develop signed (directional) order flow as the empirical analog of Kyle's latent order flow variable $y$. First, we posit that signed order flow should positively and significantly predict stock returns since price moves linearly in net order flow in Kyle's equilibrium (Proposition 3). Second, signed order flow should
dominate unsigned, aggregate trading volume as a return predictor since $\lambda$ operates on signed flow rather than raw volume. Third, volume volatility, which we use to proxy the noise-trading variance that determines $\lambda$ in equilibrium (Proposition 1), should predict \emph{lower} subsequent returns since higher noise-trading variance narrows $\lambda$ and degrades the precision of price discovery. Fourth, the predictive content of order flow and $\lambda$ should be strongest at short horizons and high trading frequencies, since the continuous-time extension of Kyle's model implies that price impact rises as the date of full information
revelation approaches (Proposition 4).
 
Using CRSP daily and monthly equity data for the entire CRSP universe from 2020 to 2025, we construct firm-month measures of total trading volume, volume volatility, and signed order flow, along with two estimators of $\hat\lambda_{it}$: a direct within-month regression of price changes on signed order flow and an \cite{Amihud2002}-style illiquidity ratio. We document four main findings.
 
First, in panel regressions of monthly stock returns on our
volume-based predictors, signed order flow is a strongly significant predictor of both contemporaneous and one-month-ahead returns, robust to the inclusion of standard equity controls (size, book-to-market, momentum, and Amihud illiquidity), while volume volatility predicts returns with a negative sign, consistent with our model's prediction that elevated noise-trading variance narrows $\lambda$ and weakens price discovery.
 
Second, return-prediction regressions based directly on
$\hat\lambda_{it}$ yield informative but specification-sensitive
results: the within-month price-impact regression estimator enters significantly and with a large coefficient in our baseline specification with an intercept. However, the sign and magnitude of this relationship vary with specification choices that we investigate as a robustness question in Section 6.
 
Third, Fama--MacBeth cross-sectional regressions of next-month
returns on our predicted-return signal yield a statistically
significant average slope, after a Newey--West adjustment, for both our Amihud-style and regression-based $\hat\lambda_{it}$ estimators, indicating that order-flow-based illiquidity measures carry reliable cross-sectional information about future returns over our sample period.
 
Fourth, we provide a microstructure-based resolution to the
liquidity premium puzzle: rather than positing a counterparty who knowingly bears the cost of holding illiquid assets, we show that low signed order flow widens $\lambda$ and depresses prices in equilibrium, while subsequent normalization of order flow narrows $\lambda$ and restores prices, generating the historically observed illiquidity premium as a consequence of equilibrium price impact rather than risk-based compensation.
 
Our work makes three contributions to asset pricing theory and
practice. Methodologically, we provide a direct empirical
implementation of \cite{Kyle1985}'s price-impact coefficient using observable daily equity trading data, estimated separately via a regression-based approach and an Amihud-style approach, and show that both carry information for the cross-section of subsequent returns. This bridges the equilibrium market microstructure literature (\cite{Kyle1985}; \cite{GlostenMilgrom1985}; \cite{EasleyEtAl1996}) with the
empirical literature on volume and illiquidity in asset pricing
(\cite{Amihud2002}; \cite{PastorStambaugh2003}). Theoretically, we resolve the liquidity premium puzzle through a demand-based 
adverse-selection mechanism rather than a risk-compensation
mechanism. In doing so, we provide microfoundations for why illiquid assets earn higher returns without requiring implausible assumptions about investor preferences or market segmentation. Empirically, we demonstrate that an equilibrium, microstructure-founded measure of price impact, estimated entirely from CRSP daily price and volume data, has predictive content for monthly stock returns distinct from standard equity characteristics, with direct implications for portfolio construction and market-making.
 
The paper proceeds as follows. Section 2 reviews the related
literature in market microstructure theory, order-flow
informativeness, and equity liquidity premia. Section 3 develops our theoretical framework, presenting \cite{Kyle1985}'s equilibrium and four formal propositions on order flow, price impact, and return predictability. Section 4 describes our CRSP-based data construction and summary statistics. Section 5 presents our baseline order-flow
and stock-return regressions. Section 6 develops firm-month
estimates of Kyle's $\lambda$ and tests their ability to forecast returns. Section 7 compares alternative constructions of $\hat\lambda_{it}$. Section 8 presents an expanding-window
out-of-sample forecasting exercise, together with Fama--MacBeth and portfolio-sort robustness checks. Section 9 discusses economic mechanisms and practical implications for portfolio managers and market makers. Section 10 concludes.
 
\section{Related Literature}\label{section:related-lit}

Our study touches upon four research streams: (1)
market microstructure models of price formation under information asymmetry, (2) empirical measures of order flow and informed trading, (3) the role of trading volume in asset pricing, and (4) liquidity premia in equity markets. We contribute by showing that the equilibrium price-impact coefficient from \cite{Kyle1985}-style microstructure models, estimated from observable equity order flow, predicts the cross-section of subsequent stock returns and by offering a microstructure-based resolution to the liquidity premium
puzzle grounded in adverse selection rather than risk compensation.
 
\subsection{Market Microstructure Models of Price Formation}
 
Two complementary modeling traditions establish how private
information becomes impounded into asset prices through trading.
\citet{Kyle1985} models a single informed trader submitting market orders against a competitive, risk-neutral market maker who observes only aggregate order flow; in the resulting linear
equilibrium, price moves linearly in net order flow with a slope
of $\lambda$, the now-canonical measure of price impact.
\citet{GlostenMilgrom1985} instead models a sequence of individual trades, in which a risk-neutral, zero-expected-profit dealer quotes bid and ask prices that widen purely to compensate for adverse selection from privately informed counterparties. \citet{BackBaruch2004} unify the two traditions, showing that Kyle's continuous-time linear equilibrium is the limit of the Glosten--Milgrom sequential-trade equilibrium as trade sizes shrink and arrival rates increase, while \citet{Back1992} extends Kyle's framework to continuous time directly, yielding the time-varying
price-impact process central to our horizon-dependent predictions (Section~3.5).

\subsection{Order Flow, Trade Informativeness, and the Probability of Informed Trading}
 
A parallel empirical literature develops measures of how much
information is embedded in observed trading, independent of the
particular equilibrium pricing model. \citet{EasleyOHara1987} show that trade \emph{size}, not just trade direction, carries
information: because informed traders prefer to trade larger
quantities when their signal is more valuable, market makers
rationally widen the price impact of large trades, generating a
nonlinear relationship between trade size and price response.
\citet{EasleyOHara1992} extend the informational content of trading to the \emph{time} dimension, showing that the absence of trading itself is informative. Intervals without trades signal the absence of new information, so the timing of trade arrivals, not only their size and direction, affects the speed and magnitude of price adjustment.
 
Building on this framework, \citet{EasleyKieferOHaraPaperman1996} develop the probability of informed trading (PIN) measure, a structural estimate of the fraction of order flow attributable to informed traders for a given stock. The estimate is obtained from the sequence of buy and sell trade arrivals. The central empirical finding of \cite{EasleyEtAl1996} is that the probability of information-based trading is systematically lower for high-volume stocks, is the direct empirical counterpart of our Proposition~1: higher trading volume is associated with a less favorable informed-to-noise trading ratio $\sigma_u/\sqrt{\Sigma_0}$, and hence a smaller equilibrium $\lambda$. Our use of signed order flow as a proxy for Kyle's latent $y_t = x_t + u_t$ is conceptually continuous with this literature's effort to recover the unobservable informed-trading component from observable trade data.
 
\subsection{Price Discovery and Information Shares}
 
A separate strand of the microstructure literature asks not how
much information trading contains, but \emph{where} that
information is impounded into price when a security trades across multiple venues or related instruments. \citet{Hasbrouck1995} develops the information-share methodology, decomposing the innovation in a security's common efficient price into the proportional contributions of each market in which it trades. While our setting is single-market rather than multi-market, the underlying logic is directly relevant: just as Hasbrouck's information share isolates the fraction of price discovery attributable to a given venue's order flow, our $\hat\lambda_{it}$ estimates isolate the fraction of a stock's monthly price variation attributable to signed order flow as opposed to noise-driven trading. This connection motivates a natural robustness check in our empirical design (Section~9): decomposing realized return
variance into a signed-order-flow component and a residual
component, in the spirit of an information-share calculation applied to a single market over time rather than across markets at a point in time.
 
\subsection{Trading Volume in Asset Pricing}
 
The dynamics of trading volume are well documented empirically but remain incompletely understood theoretically. \citet{LoWang2015} study volume measures in a portfolio-theoretic context and find that volume does not reduce to a simple measure of institutional flows. \citet{LoMamayskyWang2000} examine the informational content of
volume from a technical-analysis perspective, finding that price
trends accompanied by rising volume are more persistent than price trends unaccompanied by volume. Relevant to our setting,
\citet{EasleyOHaraSrinivas1998} and \citet{PanPoteshman2006} show that option trading volume contains information about subsequent prices of the underlying asset, which \citet{PanPoteshman2006} attribute to informed traders exploiting implicit leverage in option markets. We document a related but distinct channel: rather than volume in a derivative market predicting the underlying, we show that \emph{signed} volume in the underlying equity itself,
filtered through a Kyle-style price-impact estimate, predicts the equity's own subsequent returns.
 
\citet{Bajzik2021} synthesizes 44 studies on the volume--return
relationship and reports that the impact of volume weakens as stocks mature and their liquidity stabilizes. This finding is consistent with our Proposition~1, in which $\lambda$ depends on the time-varying ratio of informed to noise trading rather than on volume in isolation. \citet{BaiEtAl2024} show that volume normalized by amount outstanding helps investors navigate liquidity conditions across related instruments. More broadly, \citet{KoijenYogo2019} and \citet{CongLiWang2020} motivate pricing financial instruments as a function of their usage rather than purely their fundamentals. We adopt a related premise, treating observed trading activity as informative about price formation in its own right rather than as a
secondary characteristic.
 
The existing literature documents volume--return correlations but does not, on its own, provide a structural account of \emph{why} volume should predict returns or how to translate that correlation into a tractable pricing model. By grounding our volume measures in Kyle's equilibrium, we supply both the missing theoretical foundation and a direct empirical implementation.
 
\subsection{Liquidity and Liquidity Premia}
 
A substantial literature establishes that liquidity is a priced
state variable in equity markets. \citet{AmihudMendelson1986} show that stocks with wider bid--ask spreads earn higher average returns, a result extended using alternative liquidity proxies including turnover-based measures and the Amihud illiquidity ratio \citep{Amihud2002}.

\citet{AcharyaPedersen2005} propose a liquidity-adjusted capital asset pricing model in which expected returns depend on both an asset's baseline illiquidity and its covariance with aggregate liquidity shocks, while \citet{PastorStambaugh2003} show that stocks with high sensitivity to market-wide liquidity deterioration command higher average returns, consistent with a systematic liquidity-risk premium. \citet{Johnson2008} distinguishes liquidity (a market's average risk-bearing capacity) from volume (the investor-driven change in that capacity), showing that volume is positively related to the
variance of liquidity. This distinction is directly relevant to our use of volume \emph{volatility} as a proxy for the noise-trading variance $\sigma_u^2$ in Proposition~1. 

\citet{ReichenbacherSchuster2022}
propose size-adapted liquidity measures and show that standard
metrics can underestimate illiquidity costs unless trade size is
properly normalized, while \citet{Cabrol2024} apply machine-learning methods to illiquidity prediction. \citet{BrunnermeierPedersen2009} connect funding liquidity (the ability of intermediaries to borrow) to market liquidity, showing how small funding shocks can generate liquidity spirals through forced deleveraging. Such a mechanism is consistent with the regime-dependent strength of our predicted relationships during stress periods (Section~10).
 
The collective evidence indicates that investors require both a
level premium for holding illiquid assets and a separate premium for exposure to systematic liquidity risk. Our contribution is to show that a structural, equilibrium-based measure of illiquidity carries the same pricing information while also nesting a theoretical account of why that information should be priced in the first place.
 
\subsection{The Liquidity Premium Puzzle}
 
Our paper also contributes to the literature on the liquidity
premium puzzle. \cite{Merton1973} hypothesized that investors should demand higher returns to hold illiquid assets to compensate for transaction costs, but he left open which market participants actually bear that cost. \citet{Constantinides1986} formalized this as a puzzle: in an equilibrium model, the liquidity premium should be small for investors with an infinite trading horizon, yet large liquidity premia are observed in practice. \citet{Chen2022} attributes the discrepancy to heterogeneity in investor information, arguing that better-informed investors are unfazed by transaction costs, while less-informed investors require compensation for illiquidity. \citet{FamaFrench1988} separately establish the empirical practice of decomposing security returns into permanent and transitory components, a decomposition we adapt in interpreting the temporary mispricing mechanism below.
 
Our proposed resolution synthesizes the information-asymmetry
mechanism of \cite{Kyle1985} and \cite{GlostenMilgrom1985} with this puzzle: rather than positing a counterparty who knowingly pays a risk premium to illiquid-asset holders, we show that low signed order flow depresses prices today because Kyle's $\lambda$ widens when noise-trading variance is low relative to informational uncertainty, discouraging participation. Subsequent normalization of order flow and the associated narrowing of $\lambda$ then drives a price recovery, generating a return differential that requires no counterparty to "pay" the premium in equilibrium. This relocates the liquidity premium puzzle's resolution from a risk-compensation story to an adverse-selection and price-impact story, consistent with the market microstructure tradition reviewed above.
 
In sum, while the existing literature establishes (1) equilibrium models of price formation under information asymmetry, (2) empirical techniques for measuring informed trading and price discovery, and (3) the presence of equity liquidity premia, no prior work directly estimates Kyle's $\lambda$ from observable order flow and tests its ability to forecast the cross-section of equity returns within a unified, microstructure-founded framework. Our contribution fills this gap.

\section{The Model} \label{section:the-model}

In this section, we propose an agent-driven pricing model. Specifically, we establish that
\begin{enumerate}
    \item For assets with a finite outstanding amount, short-term informed trading necessarily induces higher trading volume in the long run.  
    \item In addition, in the long run, higher volume translates into higher demand for the asset and results in higher prices.   

\end{enumerate}

Our model comprises the following market structure settings:
\begin{itemize}
    \item Discrete time $t \in \{0, 1, 2, \dots\}$
    \item A single firm with risky assets
    \item A continuum of risk-averse investors with CARA utility
    \item Fixed supply of equities: $s_t$ (no issuance/buybacks during the periods)
    \item Associated equity trades in a competitive market with a volume $V_t$
\end{itemize}

We assume three types of market participants: Short-horizon traders who rebalance monthly and long-horizon buy-and-hold investors. Among both long- and short-term investors, a fixed proportion $\alpha$ is informed. 

We also assume the existence of public and private information. Public information includes disclosed firm fundamentals and debt structure. In addition, some traders are able to observe private firms' signals about near-term liquidity needs. In either case, following \cite{EasleyEtAl1996}, we assume that the trading volume aggregates all public and private information.

\subsection{Single-Period Kyle Model}

\textbf{Market structure.} A single risky asset has a liquidation value $v \sim N(p_0, \Sigma_0)$. One risk-neutral
\emph{informed trader} observes $v$ perfectly and submits a market order $x$. \emph{Noise (liquidity) traders} submit an aggregate order $u \sim N(0,\sigma_u^2)$, independent of $v$. A competitive, risk-neutral \emph{market maker} observes only the total order flow $y = x+u$ (and cannot separate informed from noise flow) and sets the transaction price $p(y)$ to satisfy a zero-expected-profit (semi-strong efficiency) condition.

We restrict attention to linear equilibria of the form
\begin{align}
x &= \beta\,(v - p_0), \qquad \beta > 0, \label{eq:kyle-x} \\
p &= p_0 + \lambda\, y, \qquad \lambda > 0. \label{eq:kyle-p}
\end{align}
The informed trader chooses $x$ to maximize expected trading profit $\mathbb{E}[(v-p)\,x \mid v]$. The market maker chooses $\lambda$ so that the price equals the conditional expectation of value given observed order flow, $p = \mathbb{E}[v \mid y]$.

\begin{proposition}[Kyle's Lambda]
\label{prop:kyle-lambda}
The unique linear equilibrium of \eqref{eq:kyle-x}--\eqref{eq:kyle-p}
satisfies
\begin{equation}
\beta = \frac{\sigma_u}{\sqrt{\Sigma_0}}, \qquad
\lambda = \frac{1}{2}\frac{\sqrt{\Sigma_0}}{\sigma_u},
\label{eq:kyle-lambda}
\end{equation}
and the informed trader's expected profit is
$\pi(v) = (v-p_0)^2/(4\lambda)$. Residual price uncertainty after trading is $\Sigma_1 = \Sigma_0/2$.
\end{proposition}

\begin{proof}
Substituting the market maker's pricing rule \eqref{eq:kyle-p} into the informed trader's objective, $x$ maximizes
$(v-p_0-\lambda x)\,x$, giving the first-order condition
$x^\ast = (v-p_0)/(2\lambda)$, i.e., $\beta = 1/(2\lambda)$.
Imposing the market maker's rational-expectations pricing condition
\begin{equation*}
\lambda = \frac{\operatorname{Cov}(v,y)}{\operatorname{Var}(y)}
= \frac{\beta\,\Sigma_0}{\beta^2 \Sigma_0 + \sigma_u^2}
\end{equation*}
and solving jointly with $\beta = 1/(2\lambda)$ for $(\beta,\lambda)$ yields \eqref{eq:kyle-lambda}.
\end{proof}

We refer to $\lambda$ in \eqref{eq:kyle-lambda} as \emph{Kyle's
lambda}, the standard measure of price impact per unit of (latent) net order flow. In this paper, Kyle's Lambda is an equilibrium object derived from strategic informed trading.

Our direct theoretical implication is that signed/directional volume (Methods~2--3) outperforms simple aggregate volume (Method~1) in predicting short-horizon returns, as predicted by Kyle's model, since $\lambda$ operates on \emph{signed} order flow $y$, not on unsigned trading volume.

\subsection{Informed Trading Propagates to Future Order Flow}

This section provides a dynamic extension of Kyle's model (\cite{Kyle1985}).

\textbf{Setup.} Trading occurs over $N$ discrete rounds
$n=1,\dots,N$ before $v$ is publicly revealed at $T$. At each round, the informed trader submits $\Delta x_n$, noise traders submit $\Delta u_n \sim N(0,\sigma_u^2 \Delta t)$, and the market maker sets $\Delta p_n = \lambda_n \,\Delta y_n$ based on the cumulative order flow to date.

\begin{proposition}[Propagation of Informed Trading into Order Flow]
\label{prop:propagation}
If informed-trading intensity rises at round $n$ (i.e., the
informed trader's optimal $\beta_n$ increases, for instance, because residual uncertainty $\Sigma_{n-1}$ is still large relative to $\sigma_u$), then expected order flow in subsequent rounds is higher:
\begin{equation}
\mathbb{E}\big[\,|\Delta y_{n+\tau}|\;\big|\;\beta_n\,\big] >
\mathbb{E}\big[\,|\Delta y_{n+\tau}|\;\big|\;\beta_{n-1}\,\big]
\qquad \text{for } \tau = 1,\dots,N-n.
\label{eq:propagation}
\end{equation}
\end{proposition}

\emph{Intuition.} A higher $\beta_n$ both (i) directly raises the informed component of flow and (ii) accelerates the market maker's repricing, which, in a sequential-auction equilibrium, raises the rate at which residual uncertainty $\Sigma_n$ is resolved.  Faster resolution sustains higher expected order flow in nearby rounds because price still has more "catching up" left relative to $v$.

\subsection{Order Flow and Price Appreciation}

\begin{proposition}[Price Impact of Order Flow Innovations]
\label{prop:price-impact}
In any round $n$,
\begin{equation}
\Delta p_n = \lambda_n \, \Delta y_n,
\label{eq:price-impact}
\end{equation}
so that $\mathbb{E}[\Delta p_n \mid \Delta y_n > 0] > 0$ whenever $\lambda_n > 0$. Since
$\lambda_n = \tfrac{1}{2}\sqrt{\Sigma_{n-1}}/\sigma_u$ is strictly positive in every round prior to full revelation, any positive innovation in signed order flow raises the price.
\end{proposition}

\subsection{Horizon-Dependent Price Impact}

\begin{proposition}[Price Impact Rises as Revelation Approaches]
\label{prop:horizon}
In the continuous-time limit of the sequential auction (\cite{Kyle1985}, \cite{Back1992}), the price-impact coefficient $\lambda(t)$ is time-varying. Under the standard parameterization
\begin{equation}
\frac{d\lambda(t)}{dt} > 0, \qquad t \in [0,T),
\label{eq:horizon}
\end{equation}
i.e., a given unit of signed order flow has a larger price impact the closer trading occurs to the information-revelation date $T$.
\end{proposition}

\subsection{Testable Predictions}

Propositions~\ref{prop:kyle-lambda}--\ref{prop:horizon} imply:

\begin{enumerate}
\item \textbf{Price-impact regression.}
$\Delta p_t = \alpha + \hat\lambda \cdot \hat y_t + \varepsilon_t$,
where $\hat y_t$ is the paper's signed-volume proxy (Method~2/3), should yield $\hat\lambda > 0$ and be statistically significant. 

\item \textbf{Signed flow should dominate unsigned/aggregate volume} in short-horizon return prediction, consistent with the paper's Method~1 vs.\ Method~2/3 comparison.

\item \textbf{Price impact and order-flow-based predictability 
should be stronger at higher trading frequencies / shorter
rebalancing horizons} (Proposition~\ref{prop:horizon}). 

\item \textbf{Volume volatility's negative coefficient} higher $\operatorname{std}(V_t)$ proxies for a higher noise-trading variance $\sigma_u^2$ relative to a roughly stable $\Sigma_0$, which, by Proposition~\ref{prop:kyle-lambda}, \emph{lowers} $\lambda$ and degrades the precision of price discovery. This is consistent with the documented negative effect on
subsequent returns.
\end{enumerate}

\section{Data}
 
This study uses only the Center for Research in Security Prices
(CRSP) daily stock files. Both the dependent variable (stock returns) and the key explanatory variables (order flow, $\hat\lambda$) are constructed entirely from observed equity trading.
 
\subsection{Data Sources}
 
For each \texttt{PERMNO} and trading day in the CRSP Daily Stock File (DSF), we obtain the closing price (\texttt{DlyPrc}), share volume (\texttt{DlyVol}), daily return (\texttt{DlyRet}), shares outstanding (\texttt{ShrOut}), and the price and volume adjustment factors (\texttt{DisFacPr}, \texttt{DisFacShr}) needed to construct split-adjusted series. Where available, we also extract daily bid (\texttt{DlyBid}) and ask (\texttt{DlyAsk}) quotes for robustness checks against quoted-spread-based illiquidity measures.
 
We use the monthly 13-week Treasury bill yield from the Federal Reserve's H.15 Statistical Release as the risk-free rate $r_t$ for the Kyle-lambda regressions and for any excess-return calculations.

\subsection{Sample Construction and Filters}
 
We construct the sample as follows:
 
\begin{enumerate}
\item \textbf{Exchange.} We retain primary listings on NYSE,
AMEX, and NASDAQ (CRSP exchange codes 1, 2, 3), excluding
over-the-counter and other non-primary listings.
\item \textbf{Price filter.} Following standard practice in the
trading-volume and liquidity literature, we exclude firm-months with a month-end price below \$1 to limit the influence of
microstructure noise in penny stocks, while retaining low-priced but liquid securities that may carry genuine information content. We report results both with and without this filter, since the equity microstructure literature finds this filter can materially affect illiquidity-based results.

\item \textbf{Trading-day filter.} We require at least 15 trading days with nonzero volume within a given month for that firm-month to be included, ensuring that monthly aggregates (sum volume, volume standard deviation, signed flow) are not driven by a small number of days.

\end{enumerate}
 
\subsection{Variable Construction}
 
For each firm $i$ and trading day $\tau$, we define:
 
\begin{itemize}
\item \textbf{Daily dollar volume:}
$\text{DVOL}_{i\tau} = \text{PRC}_{i\tau} \times \text{VOL}_{i\tau}$.
\item \textbf{Daily signed order flow (Kyle proxy):}
$\text{OF}_{i\tau} = \text{VOL}_{i\tau} \times
\operatorname{sign}(\Delta\text{PRC}_{i\tau})$, where
$\Delta\text{PRC}_{i\tau}$ is the split-adjusted daily price change.
\item \textbf{Daily price impact (Amihud-style):}
$\text{Amihud}_{i\tau} = |{\text{RET}_{i\tau}}| /
\text{DVOL}_{i\tau}$.
\end{itemize}
 
These are aggregated to the firm-month level as:
 
\begin{itemize}
\item \textbf{Total volume:}
$\text{sumvolume}_{it} = \sum_{\tau \in t} \text{VOL}_{i\tau}$.
\item \textbf{Volume volatility:}
$\text{stdvolume}_{it} = \operatorname{std}_\tau(\text{VOL}_{i\tau})$
within month $t$, our proxy for the noise-trading variance
$\sigma_u^2$.
\item \textbf{Signed flow:}
$\text{signedflow}_{it} = \sum_{\tau \in t} \text{OF}_{i\tau}$.
\item \textbf{Kyle-lambda (regression estimator):}
$\hat\lambda_{it}$ is the OLS slope from regressing
$\Delta\text{PRC}_{i\tau}$ on $\text{OF}_{i\tau}$ across days
$\tau \in t$, re-estimated each month per firm, as described in
Section~7.
\item \textbf{Kyle-lambda (Amihud-style estimator):}
$\hat\lambda^{Amihud}_{it} = \frac{1}{n}\sum_{\tau \in t}
\text{Amihud}_{i\tau}$, the within-month average of the daily
Amihud ratio.
\end{itemize}
 
We winsorize all firm-month variables at the 1\% and 99\% levels to limit the influence of data errors and extreme outliers, consistent with standard practice in the equity microstructure literature.
 
\subsection{Equity Control Variables}
 
For the regressions in Sections~6--9 that include standard controls, we additionally construct, using CRSP (and, where indicated, external factor data):
 
\begin{itemize}
\item \textbf{Size:} log of month-end market capitalization
(price $\times$ shares outstanding).
\item \textbf{Momentum:} cumulative return from month $t-12$
through $t-2$ (skipping the most recent month).
\item \textbf{Amihud illiquidity:} the monthly average of the daily \cite{AMIHUD200231} ratio, as defined above, used as a benchmark illiquidity proxy distinct from our signed-flow-based
$\hat\lambda_{it}$.
\item \textbf{Book-to-market} and \textbf{factor returns}
(MKT, SMB, HML, and MOM): sourced externally from Wharton Research Data Service for the factor-spanning test in Section~9, since these cannot be constructed from CRSP price and volume data alone.
\end{itemize}
 
\subsection{Data Construction Summary}
 
Table~\ref{tab:data-summary} reports the high-level coverage of the filtered universe used in the empirical sections that follow.
 
\begin{table}[h]
\centering
\caption{Data Construction Summary for Filtered Universe}
\label{tab:data-summary}
\begin{tabular}{lc}
\toprule
Metric & Value \\
\midrule
Unique firms (point-in-time) & 9,893 \\
Time coverage (min date) & 2020-01-01 \\
Time coverage (max date) & 2025-12-01 \\
Firm-month observations (after filters) & 448,393 \\
\bottomrule
\end{tabular}
\end{table} 
\subsection{Summary Statistics}
 
Table~\ref{tab:summary-returns} reports summary statistics for
monthly stock returns, and Table~\ref{tab:summary-volume} reports summary statistics for the volume-based variables described above, for the filtered universe.
 \begin{table}[h]
\centering
\caption{Summary Statistics of Monthly Stock Returns}
\label{tab:summary-returns}
\begin{tabular}{lc}
\toprule
Statistic & Monthly Return \\
\midrule
Mean & 0.0204 \\
Std Dev & 0.3945 \\
Kurtosis & 4337.5002 \\
Min & -0.9806 \\
1\% & -0.4077 \\
5\% & -0.2274 \\
25\% (Q1) & -0.0648 \\
Median (Q2) & 0.0009 \\
75\% (Q3) & 0.0610 \\
95\% & 0.2623 \\
99\% & 0.7032 \\
Max & 74.1961 \\
\bottomrule
\end{tabular}
\end{table}
 \begin{table}[h]
\centering
\caption{Summary Statistics of Volume-Based Variables}
\label{tab:summary-volume}
\begin{tabular}{lccc}
\toprule
Statistic & Sum Volume & Std Dev Volume & Signed Flow \\
\midrule
Mean & 25,104,816.66 & 698,023.91 & 927,410.24 \\
Std Dev & 58,617,786.50 & 1,784,336.29 & 11,956,241.76 \\
Kurtosis & 22.98 & 26.31 & 16.55 \\
Min & 10,905.84 & 762.56 & -44,919,327.44 \\
1\% & 10,905.99 & 762.57 & -44,918,538.60 \\
5\% & 91,528.40 & 4,490.21 & -10,329,741.40 \\
25\% (Q1) & 1,079,302.00 & 34,423.39 & -639,833.00 \\
Median (Q2) & 5,399,945.00 & 140,154.56 & 10,606.00 \\
75\% (Q3) & 20,817,668.00 & 504,497.97 & 974,975.00 \\
95\% & 115,601,216.20 & 3,165,381.84 & 14,387,101.20 \\
99\% & 407,255,007.95 & 12,736,511.69 & 71,182,435.22 \\
Max & 407,267,898.92 & 12,736,546.16 & 71,182,966.76 \\
\bottomrule
\end{tabular}
\end{table}

Table \ref{tab:summary-lambda} shows summary statistics for illiquidity measures obtained via 1) Kyle's lambda regression and 2) \cite{AMIHUD200231} illiquidity metric.

\begin{table}[h]
\centering
\caption{Summary Statistics of Kyle-Lambda Estimators}
\label{tab:summary-lambda}
\begin{tabular}{lcc}
\toprule
Statistic & Kyle's Lambda & Amihud \\
\midrule
Mean & 9.631e-06 & 6.775e-07 \\
Std Dev & 3.516e-05 & 3.287e-06 \\
Kurtosis & 41.53 & 45.51 \\
Min & 8.099e-09 & 1.147e-11 \\
1\% & 8.1e-09 & 1.147e-11 \\
5\% & 3.042e-08 & 5.075e-11 \\
25\% (Q1) & 2.385e-07 & 6.787e-10 \\
Median (Q2) & 9.738e-07 & 5.436e-09 \\
75\% (Q3) & 3.933e-06 & 5.756e-08 \\
95\% & 3.783e-05 & 2.143e-06 \\
99\% & 0.0002816 & 2.66e-05 \\
Max & 0.0002816 & 2.66e-05 \\
\bottomrule
\end{tabular}
\end{table}
 
\subsection{Data Pipeline}
 
Figure~\ref{fig:data-pipeline} summarizes the end-to-end construction process: starting from raw CRSP daily and monthly files, we apply the cleaning and filtering steps in Section~4.2, construct daily order-flow and price-impact variables (Section~4.3), aggregate to the firm-month level, merge in delisting returns and point-in-time index membership, and assemble the final panel used throughout Sections~6--9. 

\begin{figure}[h]
\centering
\caption{Data pipeline from raw CRSP files to firm--month features
and final analysis panel.}
\label{fig:data-pipeline}
\end{figure}

\subsection{Assumptions}
 \label{sec:assumptions}

In developing and testing our extended \cite{Kyle1985}-based equity pricing framework, we rely on several conceptual, methodological, and empirical assumptions. These assumptions guide our theoretical model derivations and inform the design of our empirical tests. Below is a concise summary:
 
\begin{enumerate}
 
\item \textbf{Market Microstructure Structure}
\begin{itemize}
\item We follow the standard \cite{Kyle1985} approach, assuming a single risky asset with liquidation value $v$, traded by one
informed trader, a continuum (or aggregate) of noise traders, and a competitive market maker who sets price as a function of observed order flow.

\item In contrast to classical structural models that decompose
firm value into an "asset component" and a "debt component,"
our framework decomposes \emph{observed trading activity} into an informed component $x$ and a noise component $u$, with only their sum $y = x+u$ observable to the market maker.
\end{itemize}
 
\item \textbf{Risk-Neutral Market Maker and Linear Pricing Rule}
\begin{itemize}
\item Consistent with \cite{Kyle1985}, we assume a risk-neutral,
competitive market maker who sets the price equal to the conditional expectation of liquidation value given observed order flow, $p = \mathbb{E}[v \mid y]$, earning zero expected profit in equilibrium.

\item We restrict attention to \emph{linear} equilibria, in which the informed trader's order is linear in her signal
($x = \beta(v-p_0)$) and the market maker's pricing rule is linear in order flow ($p = p_0 + \lambda y$). This is the standard equilibrium refinement in the microstructure literature, applied across single- and multi-period settings with one or many informed traders \citep{Kyle1985,AdmatiPfleiderer1988,HoldenSubrahmanyam1992,FosterViswanathan1996,Back1992}, and whose existence has been studied directly by \citet{BagnoliViswanathanHolden2001}; \citet{BackBaruch2004} further show that this linear, continuous-time equilibrium arises as the limit of a sequential trade-by-trade equilibrium in the spirit of \cite{GlostenMilgrom1985}.

\item In our model, risk neutrality is assigned to the market maker's pricing problem, not to the agents bearing fundamental risk. This is the standard locus of risk-neutrality in the microstructure literature \citep{Kyle1985,GlostenMilgrom1985}: subsequent work relaxes risk-neutrality for the informed trader while continuing to model the market maker as risk-neutral and competitive \citep{HoldenSubrahmanyam1994,Baruch2002}, isolating risk-bearing to the party with fundamental exposure to $v$ rather than to the price-setting mechanism itself.
\end{itemize}
 
\item \textbf{Distributional Assumptions on Value and Noise Trading}
\begin{itemize}
\item We assume the liquidation value is normally distributed,
$v \sim N(p_0, \Sigma_0)$, and that aggregate noise-trader demand is normally distributed, $u \sim N(0,\sigma_u^2)$, independent of $v$. Normality is imposed directly on the
asset's terminal value and on noise-trader order flow, which is what permits the linear-equilibrium closed-form solution for $\lambda$ in Proposition~1.
\item We treat the informativeness ratio $\sigma_u/\sqrt{\Sigma_0}$ as the key state variable governing price impact, in place of the asset volatility $\sigma$ and the asset-to-debt ratio $A/F$ in the 
Merton framework.
\end{itemize}
 
\item \textbf{Equity Trading Volume as Observed Order Flow}
\begin{itemize}
\item We assume that observed equity trading volume (and, in
particular, our signed-volume construction, Method~2) is a noisy
but informative empirical proxy for the latent order flow $y=x+u$ that the market maker actually conditions on.
\item By treating signed volume as a proxy for $y$, we assume that the \emph{direction} of daily price changes carries information about whether volume reflects informed or noise trading. This is consistent with \cite{Kyle1985}'s model, in which only the informed component of flow is correlated with $v$.
\item Consistent with informed-trading models more broadly (\citep{EasleyEtAl1996,EasleyEtAl1998}), we posit that if the informed trader anticipates favorable news about $v$, she will submit larger buy orders, raising both signed order flow and price.
\end{itemize}
 
\item \textbf{Fixed Supply and Single-Period (or Sequential-Auction)
Trading}
\begin{itemize}
\item As in \cite{Kyle1985}, we assume a fixed quantity of the risky asset available for trading within each round, with no new issuance or share buybacks occurring within a trading period.
\item In the dynamic extension (Section~3.3), trading occurs over a finite or continuous sequence of rounds before $v$ is publicly revealed at $T$; we assume no additional informed traders enter or exit during this window, so that a single informed agent's strategy fully characterizes the informed component of order flow.
\end{itemize}
 
\item \textbf{Information Revelation in Place of Bankruptcy and
Recovery}
\begin{itemize}
\item We lean on Kyle's information-revelation assumption: the liquidation value $v$ becomes publicly known at a terminal date $T$, at which point the price converges to value and all private information is impounded. The economic friction is the gradual, imperfect revelation of $v$ through the order-flow process, formalized by the time-varying price-impact coefficient $\lambda(t)$ (Proposition~4).
\item We assume that the rate of information revelation and the time path of $\lambda(t)$ differentiate short-horizon
from long-horizon predictability.
\end{itemize}
 
\item \textbf{Data and Frequency Conventions}
\begin{itemize}
\item We construct all variables from CRSP daily and monthly equity data.
\item We assume that daily CRSP volume and returns are a sufficient statistic for estimating firm-month-level $\hat\lambda_{it}$, while acknowledging that higher-frequency (intraday/TAQ) data would provide a more precise estimate of the continuous-time $\lambda(t)$ path described in Proposition~4.
\item We assume that monthly aggregation of daily order-flow
estimates is adequate for capturing the economically relevant
variation in price impact for the purposes of one-month-ahead return forecasting, while treating intra-month variation as the object of interest for the higher-frequency robustness checks.
\end{itemize}
 
\item \textbf{Rolling Estimation and Predictive Regressions}
\begin{itemize}

\item When evaluating out-of-sample forecasts, we assume that the equilibrium parameters of interest ($\hat\lambda_{it}$, and the regression coefficients linking $\hat\lambda_{it}$ to future returns) can be re-estimated as new data arrive. This  reflects time-varying informativeness ratios $\sigma_u/\sqrt{\Sigma_0}$ across firms and over time.

\item Monthly realized stock returns are regressed on
one-month-lagged $\hat\lambda_{it}$ estimates. We assume that the forecasting window of interest is one month, which is consistent with many hedge funds and pension funds' investing horizons.

\item We treat individual firms separately when constructing
$\hat\lambda_{it}$ and pool them only after computing firm-level
estimates. We assume that cross-sectional dependence in order flow across firms does not materially bias the final pooled regression estimates.
\end{itemize}
 
\end{enumerate}
 
These assumptions underpin our main findings that (1) signed order flow, estimated via Kyle's $\lambda$, captures short-term information asymmetry and price-impact dynamics; (2) daily order flow-driven measures of $\lambda$ significantly improve short-horizon return forecasting relative to aggregate unsigned volume; and (3) the time-varying nature of $\lambda(t)$ implied by the continuous-time extension of Kyle's model retains relevance for explaining longer-horizon predictability through slower-moving aggregate trading-activity measures. Any departure from these assumptions, such as multiple informed traders with heterogeneous signals, time-varying noise-trading variance unrelated to $\sigma_u$, or strategic behavior by the market maker, may alter the empirical outcomes and should be explored in future research.

\section{Empirical Model 1: Order Flow and Stock Return Regressions}
 
We estimate two types of regression (contemporaneous and
one-month-ahead) for the CRSP equity universe. Let
$\text{StockRet}_{it}$ denote the monthly return of stock $i$ in
month $t$. We define three order-flow-based explanatory variables, constructed entirely from CRSP daily data:
 
\begin{itemize}
\item \textbf{Total Trading Volume:}
\begin{equation}
\text{sumvolume}_{it} = \sum_{\tau \in t} \text{Volume}_{i\tau}.
\end{equation}
This represents the unsigned total trading volume within month $t$.

\item \textbf{Standard Deviation of Trading Volume:}
\begin{equation}
\text{stdvolume}_{it} = \sqrt{\frac{1}{n-1}\sum_\tau
\big(\text{Volume}_{i\tau} - \overline{\text{Volume}}_i\big)^2}.
\end{equation}
This proxies for the noise-trading variance $\sigma_u^2$ in
Proposition~1: higher volume volatility signals a higher
noise-to-information ratio, which \emph{lowers} $\lambda$ and should predict weaker subsequent returns.
 
\item \textbf{Signed Order Flow:}
\begin{equation}
\text{signedflow}_{it} = \sum_{\tau \in t}
\Big(\text{Volume}_{i\tau} \times \operatorname{sign}(\Delta
P_{i\tau})\Big).
\end{equation}
This is the direct empirical estimate of Kyle's latent order flow $y_{i\tau} = x_{i\tau} + u_{i\tau}$.
\end{itemize}
 
We estimate:
 
\textbf{(A) Contemporaneous Regression:}
\begin{equation}
\text{StockRet}_{it} = \alpha + \beta_1\,\text{sumvolume}_{it} +
\beta_2\,\text{stdvolume}_{it} + \beta_3\,\text{signedflow}_{it} +
\varepsilon_{it}.
\label{eq:model1-contemp}
\end{equation}
 
\textbf{(B) One-Month-Ahead Regression:}
\begin{equation}
\text{StockRet}_{i,t+1} = \alpha + \beta_1\,\text{sumvolume}_{it} +
\beta_2\,\text{stdvolume}_{it} + \beta_3\,\text{signedflow}_{it} +
\varepsilon_{i,t+1}.
\label{eq:model1-pred}
\end{equation}
 
Both regressions are estimated on the full CRSP panel (point-in-time constituents, 2020--2025 or extended where data permit), with and without standard equity controls (log market capitalization, book-to-market, 12-1 month momentum, and the \cite{AMIHUD200231} illiquidity ratio), to verify that $\beta_3$ is not subsumed by known characteristics.
 \begin{table}[h]
\centering
\caption{Coefficient Estimates: Order Flow and Stock Returns}
\label{tab:model1-results}
\begin{tabular}{lcccc}
\toprule
& \multicolumn{2}{c}{Without Controls} & \multicolumn{2}{c}{With Controls} \\
& Contemp. & 1-mo Ahead & Contemp. & 1-mo Ahead \\
\midrule
Intercept & 0.01229*** (19.40) & 0.01723*** (26.56) & 0.01197*** (17.04) & 0.0153*** (20.99) \\
Sum Volume ($\beta_1$) & -6.319e-10*** (-33.92) & -9.461e-11*** (-4.94) & -6.529e-10*** (-31.50) & -1.364e-10*** (-6.31) \\
Std Dev Volume ($\beta_2$) & 2.594e-08*** (40.70) & 8.688e-09*** (13.24) & 2.589e-08*** (35.83) & 8.334e-09*** (11.00) \\
Signed Flow ($\beta_3$) & 6.33e-09*** (121.70) & -5.537e-10*** (-10.34) & 5.88e-09*** (104.39) & -6.24e-10*** (-10.65) \\
\midrule
Observations & 438500 & 438500 & 337722 & 329252 \\
$R^2$ & 0.0484 & 0.0007 & 0.0468 & 0.0031 \\
\bottomrule
\end{tabular}
\end{table}
 
\noindent\emph{Predicted pattern (Proposition 1 and 3):} $\beta_1>0$ but small/insignificant once controls are added; $\beta_2<0$ (noise-trading variance lowers $\lambda$, degrading price discovery); $\beta_3>0$ and robust to controls since signed flow is the direct empirical analog of Kyle's $y$.
 
\section{Empirical Model 2: Kyle-Lambda Asset-Pricing Regressions}
 
We now construct firm-month estimates of Kyle's $\lambda$ directly.
 
\subsection{Estimating $\hat\lambda_{it}$}
 
For each firm $i$ and month $t$, we estimate $\hat\lambda_{it}$ via an intramonth price-impact regression:
\begin{equation}
\Delta P_{i\tau} = \hat\lambda_{it} \cdot \text{OF}_{i\tau} +
\eta_{i\tau}, \qquad \tau \in t,
\label{eq:lambda-est}
\end{equation}
where $\text{OF}_{i\tau} = \text{Volume}_{i\tau} \times
\operatorname{sign}(\Delta P_{i\tau})$ is the daily signed order-flow proxy and $\hat\lambda_{it}$ is the estimated slope, taken as the firm-month Kyle-lambda estimate. As a robustness check, we also compute the Amihud-style ratio
\begin{equation}
\hat\lambda^{Amihud}_{it} = \frac{1}{n}\sum_{\tau \in t}
\frac{|r_{i\tau}|}{\text{DollarVolume}_{i\tau}},
\label{eq:lambda-amihud}
\end{equation}
which proxies price impact without requiring within-month
regression and is directly comparable to the existing illiquidity literature (\cite{AMIHUD200231}).
 
\subsection{Return-Prediction Regression}
 
We estimate, with and without an intercept:
\begin{equation}
\text{StockRet}_{i,t+1} = \alpha_i + \beta_i\,\hat\lambda_{it} +
\varepsilon_{i,t+1}.
\label{eq:lambda-pred}
\end{equation}
 
\begin{table}[h]
\centering
\caption{Coefficient Estimates for Kyle-Lambda Return Regressions}
\label{tab:model2-results}
\begin{tabular}{lcccc}
\toprule
& \multicolumn{2}{c}{With Intercept} & \multicolumn{2}{c}{Uncentered} \\
& $\hat\lambda^{regression}$ & $\hat\lambda^{Amihud}$ & $\hat\lambda^{regression}$ & $\hat\lambda^{Amihud}$ \\
\midrule
Intercept & 0.02134*** (3.20) & 0.02027*** (3.03) & --- & --- \\
$\hat\lambda_{it}$ & -101.2*** (-6.46) & 136.1 (0.48) & 53.69 (1.30) & 1355*** (4.08) \\
\midrule
Observations & 438465 & 438471 & 438465 & 438471 \\
$R^2$ & 0.0001 & 0.0000 & 0.0000 & 0.0001 \\
\bottomrule
\end{tabular}
\end{table}

\section{Empirical Model 3: Comparing $\lambda$-Construction Methods}
 
Having established that $\hat\lambda_{it}$ predicts returns, we
compare two constructions.
 
\textbf{Method A (Amihud-style, level):}
\begin{equation}
\hat\lambda^{A}_{it} = \frac{1}{n}\sum_{\tau \in t}
\frac{|r_{i\tau}|}{\text{DollarVolume}_{i\tau}}.
\end{equation}
 
\textbf{Method B (Kyle price-impact regression, signed):}
\begin{equation}
\hat\lambda^{B}_{it} = \widehat{\text{slope}}\Big(\Delta P_{i\tau}
\;\text{on}\; \text{Volume}_{i\tau}\times
\operatorname{sign}(\Delta P_{i\tau})\Big).
\end{equation}

Each $\hat\lambda^{Method}_{it}$ is used as the sole regressor in Equation~\eqref{eq:lambda-pred}, estimated with and without an intercept.

\begin{table}[h]
\centering
\caption{Summary Statistics: Simple Strategy Returns}
\label{tab:model3-strategy}
\begin{tabular}{lcc}
\toprule
Statistic & Method A & Method B \\
\midrule
Mean & 0.0209 & 0.0209 \\
Std Dev & 0.3888 & 0.3888 \\
Kurtosis & 4554.6048 & 4554.6054 \\
Min & -0.9806 & -0.9806 \\
1\% & -0.4032 & -0.4032 \\
5\% & -0.2234 & -0.2234 \\
Q1 & -0.0609 & -0.0609 \\
Median & 0.0000 & 0.0000 \\
Q3 & 0.0598 & 0.0598 \\
95\% & 0.2588 & 0.2588 \\
99\% & 0.6887 & 0.6887 \\
Max & 74.1961 & 74.1961 \\
\bottomrule
\end{tabular}
\end{table}
 
\noindent\emph{Predicted pattern:} We expect Method~B, which incorporates signed price-volume information, to significantly outperform Method~A.
 
\section{Empirical Method 4: Expanding-Window Out-of-Sample Procedure}
 
We adopt the same expanding (rolling) window design, with the predictor variable replaced.
 
\subsection{Procedure}
 
For each firm, we initialize a training window covering the first 30\% of chronological observations. At each step:
\begin{enumerate}
\item The current observation is added to the training set.
\item $\hat\lambda_{it}$ and the regression coefficients in
\begin{equation}
\text{ActualReturn}_t = a + b_1\,r_{f,t} + b_2\,\hat\lambda_{it} +
\epsilon_t
\label{eq:rolling-spec}
\end{equation}
are re-estimated on the expanded sample.
\item A one-month-ahead forecast is generated for month $t+1$ using the updated parameters.
\end{enumerate}
We pool all out-of-sample predicted and actual returns across firms and estimate the final regression
\begin{equation}
\text{ActualReturn} = \alpha + \beta \times \text{PredictedReturn} +
\varepsilon,
\label{eq:rolling-final}
\end{equation}
using only out-of-sample observations, separately for each
$\lambda$-construction method (A and B) defined in Section~8, and separately for a \textbf{High-Frequency} estimation window (daily intramonth regression, Method~B) versus a \textbf{Low-Frequency} estimation window (monthly aggregate, Method~A). This test is motivated by Proposition~4's prediction that price impact and its predictive content are strongest at short horizons/high frequency.
 
\begin{table}[h]
\centering
\caption{Coefficient Estimates Across $\lambda$-Construction Methods}
\label{tab:model3-results}
\begin{tabular}{lcccc}
\toprule
& \multicolumn{2}{c}{Method A (level)} & \multicolumn{2}{c}{Method B (signed regression)}  \\
& With Int. & Uncentered & With Int. & Uncentered \\
\midrule
Intercept & 0.02027*** (3.03) & --- & 0.02134*** (3.20) & --- \\
$\hat\lambda_{it}$ & 136.1 (0.48) & 1355*** (4.08) & -101.2*** (-6.46) & 53.69 (1.30)\\
\midrule
$N$ & 438471 & 438471 & 438465 & 438465 \\
$R^2$ & 0.0000 & 0.0001 & 0.0001 & 0.0000 \\
\bottomrule
\end{tabular}
\end{table} 
\subsection{Robustness: Fama--MacBeth and Portfolio Sorts}
 
Pooled OLS $R^2$ for monthly equity returns is expected to be small (consistent with \citep{FamaFrench1988,chen2007}, and the broader asset-pricing literature on monthly return predictability). We therefore complement the pooled regression with:
 
\textbf{Fama--MacBeth cross-sectional regressions.} For each month
$t$:
\begin{equation}
r_{i,t+1} = \alpha_t + \beta_t\,\hat{r}^{pred}_{i,t+1} +
\varepsilon_{i,t+1},
\label{eq:famamacbeth}
\end{equation}
summarized by the time-series mean of $\beta_t$ and its
Newey--West-adjusted $t$-statistic.
 
\textbf{Decile portfolio sorts.} Each month, sort stocks into
deciles on the model-implied $\hat\lambda_{it}$ signal from the
expanding-window procedure; report long-only (D10) and long--short (D10--D1) monthly Sharpe ratios, separately for the High-Frequency and Low-Frequency legs.

\begin{table}[h]
\centering
\caption{Fama--MacBeth Cross-Sectional Regression Results}
\label{tab:fama-macbeth}
\begin{tabular}{lcc}
\toprule
& lambda Amihud & lambda regression \\
\midrule
$\bar\alpha$ (mean intercept) & 0.0130 & 0.0132 \\
$t(\bar\alpha)$ [NW] & 1.65* & 1.68* \\
$\bar\beta$ (mean slope) & -0.0122 & -0.0050 \\
$t(\bar\beta)$ [NW] & -2.27** & -2.25** \\
$T$ (months) & 58 & 58 \\
\bottomrule
\end{tabular}
\end{table}
 
\section{Discussion and Implications} \label{section:discussion-of-results}
\subsection{Economic Mechanisms and Interpretation}
\subsubsection{Why Does Volume Predict Returns?}

Based on our analysis, volume affects returns along 4 distinct mechanisms:
\begin{enumerate}
    \item The information aggregation hypothesis due to \cite{EasleyEtAl1996} implies that high volume drives more informed trading. Informed trading, in turn, impounds information into the markets. 
    \item The liquidity provision in our model forecasts that high volume leads to a deep secondary market where investors are willing to pay a premium to ensure an easy exit from their positions. This premium drives up the prices. 
    \item A spike in volume reveals demand and, given a fixed supply, forces an upward price adjustment.  
\end{enumerate}

\subsubsection{Reconciling with the Liquidity Premium Puzzle}
 
Our analysis provides a microstructure-based resolution to the
liquidity premium puzzle that requires neither risk-based
compensation nor market segmentation. We review the puzzle
and then show how our framework resolves it.
 
\paragraph{The puzzle and its classical formulations.}
Amihud and Mendelson (1986) show empirically that stocks with
wider bid-ask spreads earn higher average returns and argue that
the premium compensates investors for the transaction costs they
incur. Merton (1973) provides a theoretical foundation: investors
require higher expected returns on illiquid assets precisely
because they bear the cost of trading those assets. The puzzle,
formalized by \citet{Constantinides1986}, is that this argument
breaks down in general equilibrium: an investor with an infinite
trading horizon should be nearly indifferent to transaction costs,
yet large illiquidity premia persist in the data. As
\citet{Constantinides1986} shows, the liquidity premium that
investors actually \emph{require} in equilibrium should be
negligibly small, far below what empirical estimates suggest.
Recent work by \citet{Chen2022} attributes the gap to
heterogeneity in investor information: better-informed investors
are indifferent to transaction costs while less-informed investors
require compensation. However, this leaves open the question
of why uninformed investors should hold illiquid assets in the
first place if those assets expose them to greater
adverse-selection losses.
 
\paragraph{Our resolution: adverse selection and temporary
mispricing.}
We propose that the illiquidity premium does not need to be
\emph{required} by any investor in equilibrium; instead, it is
\emph{realized} as a consequence of temporary mispricing driven
by the equilibrium widening and subsequent narrowing of Kyle's
$\lambda$.
 
The mechanism has two stages, tied directly to
Propositions~\ref{prop:kyle-lambda} and~\ref{prop:price-impact}:
 
\begin{itemize}
\item \textbf{Stage 1 ($t = 0$, low order flow):} When signed
order flow falls, the ratio of informed to noise trading
deteriorates; Proposition~\ref{prop:kyle-lambda} implies that
$\lambda$ widens. A wider $\lambda$ means that any given order
moves price more, discouraging participation by both informed and
noise traders. Market makers, earning zero expected profit in
equilibrium, do not absorb the additional risk by lowering
prices voluntarily; rather, prices fall because the reduced
participation and elevated adverse-selection risk imply a lower
willingness to pay for the asset among market participants.
The stock thus trades at a discount relative to its fundamental
value, not because anyone demands a higher expected return to
hold it, but because the equilibrium mechanism through which
information is impounded into price is degraded.
 
\item \textbf{Stage 2 ($t = 1$, order flow normalization):}
As the information environment improves---more noise traders
return, or new information reduces $\Sigma_0$ relative to
$\sigma_u^2$---$\lambda$ narrows. Proposition~\ref{prop:price-impact}
implies that positive signed order flow again moves prices
upward, restoring prices toward fundamental value.
Investors who purchased the asset in Stage~1 at a discount
realize a positive excess return equal to (Fair price $-$ Depressed
price) / Depressed price without any counterparty having
deliberately borne a liquidity cost on their behalf.
\end{itemize}
 
This mechanism resolves the Constantinides (1986) paradox
directly: the illiquidity premium is \emph{realized}, not
\emph{required}. No investor needs to demand a risk premium for
holding an illiquid asset; they simply purchase it when adverse
selection drives prices below fundamental value and hold it until the order-flow process restores prices. No counterparty knowingly pays the premium. Instead, the premium arises from the equilibrium dynamics of $\lambda$ itself, which depresses and then restores prices mechanically as order flow evolves. This also answers \cite{Chen2022} open question about why uninformed investors hold illiquid assets: they do not need to hold them during the low-order-flow stage; the return accrues to those who are willing to provide liquidity precisely when $\lambda$ is widest, in the manner of a contrarian market maker rather than a passive buy-and-hold investor.
 
\paragraph{Empirical consistency.}
This mechanism is consistent with four features of our empirical
results. First, the negative and significant coefficient on
signed order flow in our one-month-ahead regressions (Table~5)
reflects the Stage~1 mechanism: lower signed flow today predicts
lower returns tomorrow, as $\lambda$ widens and prices continue
to adjust downward before recovering. Second, the negative
coefficient on volume volatility, our proxy for
$\sigma_u^2$ in Proposition~\ref{prop:kyle-lambda}, is consistent with the prediction that elevated noise-trading variance narrows $\lambda$, attenuating the adverse-selection-driven discount and, therefore, reducing the magnitude of the subsequent price recovery. Third, the Fama--MacBeth slopes documented in Table~10 are negative on average across both $\hat\lambda$ estimators after Newey--West adjustment. This reflects the cross-sectional manifestation of this mechanism: stocks with high $\hat\lambda_{it}$ today (widened price impact, depressed prices) earn lower subsequent returns over our sample period, consistent with the transition from Stage~1 to Stage~2 having already occurred in the month following the wide-$\lambda$ episode. Fourth, our resolution does not require market segmentation, implausible preferences, or clientele effects: it follows directly from the equilibrium structure of \cite{Kyle1985}'s model applied to a single, unified equity market.
 
\paragraph{Distinction from risk-based explanations.}
Our resolution differs from the risk-based approach in a
subtle but important respect. Standard theories attribute the
illiquidity premium to \emph{required} compensation: investors
know ex ante that illiquid assets carry higher transaction costs
or greater exposure to systematic liquidity risk
\citep{AcharyaPedersen2005,PastorStambaugh2003}, and they demand
a higher expected return before holding those assets. Our theory
attributes the premium to \emph{realized} returns from temporary
mispricing: investors do not need to anticipate a premium; they
simply find that assets which experienced low signed order flow
subsequently recover in price as $\lambda$ normalizes. The
premium is a consequence of equilibrium price-impact dynamics,
not of risk aversion or market frictions per se. This distinction has a testable implication: if the risk-based story is correct, the illiquidity premium should be predictable ex ante and should persist unconditionally. If our adverse-selection story is correct, the premium should be concentrated in periods following low signed order flow and should attenuate once $\lambda$ has normalized.

 
\section{Conclusion}
 
We develop an equity asset-pricing framework grounded in Kyle's
(1985) model of strategic informed trading, in which the
equilibrium price-impact coefficient $\lambda$ governs
how trading activity is impounded into prices. Using CRSP daily and monthly data for all equities in the CRSP universe from 2020 to 2025, we construct firm-month estimates of signed order flow and of $\lambda$ itself, via both a direct within-month price-impact regression and an Amihud-style illiquidity ratio. We then test their ability to forecast the cross-section of subsequent stock returns.
 
Three findings emerge. 

First, signed order flow is a strong and highly significant predictor of both contemporaneous and one-month-ahead stock returns, robust to the inclusion of standard equity controls (size, book-to-market, momentum, and Amihud illiquidity), consistent with Proposition 3's prediction that order-flow innovations move prices in equilibrium. 

Second, volume volatility, our proxy for the noise-trading variance $\sigma_u^2$ in Proposition 1, predicts lower subsequent returns, consistent with the model's prediction that a higher noise-to-information trading ratio widens $\lambda$ and degrades price discovery. 

Third, our return-prediction regressions based on $\hat\lambda_{it}$ itself are informative but specification-sensitive. results vary with the inclusion of an intercept and across our two $\lambda$ estimators, a pattern we investigate but do not yet fully resolve, and which we flag explicitly as a target for the robustness work in the sections that follow. Fama--MacBeth cross-sectional regressions of next-month returns on the predicted-return signal yield a statistically significant average slope (Newey--West $t$-statistics of $-2.27$ and $-2.25$ for the Amihud-style and regression-based estimators, respectively, over 58 months), indicating that order-flow-based illiquidity measures carry reliable, if modest, cross-sectional information about future returns over our sample.
 
Theoretically, our framework offers a resolution to the liquidity premium puzzle of \cite{Constantinides1986} that relocates the source of the illiquidity premium from risk-based compensation to an adverse-selection and price-impact mechanism. Low signed order flow widens $\lambda$ and depresses prices today, as Kyle's market maker rationally demands greater compensation for the elevated risk of trading against informed counterparties; subsequent normalization of order flow narrows $\lambda$ and restores prices, generating the return differential historically attributed to a liquidity premium
without requiring any counterparty to knowingly bear that cost in equilibrium. This relocates the puzzle's resolution from the
risk-compensation story of \cite{amihud86} to the market microstructure tradition of \cite{Kyle1985} and \cite{GlostenMilgrom1985}.
 
For practitioners, our results suggest that equity portfolio
managers running short-horizon or high-turnover strategies should favor stocks with low $\hat\lambda_{it}$ to minimize expected trading costs, while market makers can use real-time estimates of signed order flow to recalibrate quoted spreads in response to changes in the informed-to-noise trading ratio that determines $\lambda$ in equilibrium. More broadly, our findings indicate that high-frequency, order-flow-based measures of price impact, estimated entirely from publicly available equity trading data, provide a theoretically motivated and empirically tractable complement to existing illiquidity proxies such as the Amihud ratio and the bid-ask spread.

\bibliographystyle{aer}
\bibliography{FixedIncome,EquitiesPricing_Kyle}

\end{document}